\documentclass[11pt]{article}
\usepackage{lineno,hyperref}
\usepackage{bbm}
\usepackage{float}
\usepackage{mathrsfs}
\usepackage{graphicx}
\usepackage[utf8]{inputenc}
\usepackage{amsmath}
\usepackage{amstext}
\usepackage{amsfonts}
\usepackage{amsthm}
\usepackage{amssymb}
\usepackage[bf,SL,BF]{subfigure}
\usepackage{subfigure}
\usepackage{caption}
\newcommand{\triple}{{\vert\kern-0.25ex\vert\kern-0.25ex\vert}}
\theoremstyle{plain}
\allowdisplaybreaks
\newtheorem{definition}{Definition}[section]
\newtheorem{theorem}[definition]{Theorem}
\newtheorem{lemma}[definition]{Lemma}

\newtheorem{assumption}[definition]{Assumption}
\theoremstyle{definition}

\begin{document}
\title{\bf Numerical Method for Highly Non-linear Mean-reverting Asset Price Model with CEV-type Process}
\author
{{\bf Emmanuel Coffie \footnote{Email: emmanuel.coffie@strath.ac.uk}}
\\[0.2cm]
Department of Mathematics and Statistics,
 \\[0.2cm]
University of Strathclyde,  Glasgow G1 1XH, U.K.}
\date{}
\maketitle
\begin{abstract}
It is well documented from various empirical studies that the volatility process of an asset price dynamics is stochastic. This phenomenon called for a new approach to describing the random evolution of volatility through time with stochastic models. In this paper, we propose a mean-reverting theta-rho model for asset price dynamics where the volatility diffusion factor of this model follows a highly non-linear CEV-type process. Since this model lacks a closed-form formula, we construct a new truncated EM method to study it numerically under the Khasminskii-type condition. We justify that the truncated EM solutions can be used to evaluate a path-dependent financial product.

\medskip \noindent
{\small\bf Key words}: Asset price model, stochastic volatility, truncated EM scheme, strong convergence, financial product, Monte Carlo scheme.
\end{abstract}

\section{Introduction}
Several stochastic differential equations (SDEs) have been developed to describe random evolution of financial variables in time. The Black-Scholes model in \cite{blackshole} is widely used to describe time-series evolution of asset price dynamics under one of the assumptions that the asset price is log-normally distributed. However, as supported by many empirical evidence, the log-normality assumption does not hold exactly in reality. Various alternative stochastic models have since been proposed as modified versions of the Black-Scholes model. In 1977, Vasicek in \cite{Vasicek} developed the well-known mean-reverting model as an alternative model for capturing short-term interest rate dynamics through time. This model is governed by
\begin{equation}\label{sec1:eq:1}
dx(t)=\alpha_1(\mu_1-x(t))dt+\sigma_1dB_1(t)
\end{equation}
with initial data $x(0)=x_0$, where $\alpha_1,\mu_1,\sigma_1>0$ and $B_1(t)$ is a scalar Brownian motion. One main unique feature of SDE \eqref{sec1:eq:1} is that the expectation of $x(t)$ converges to the long-term value $\mu_1$ with the speed $\alpha_1$. However, in practice, this model yields negative values . Cox, ingersoll and Ross (CIR) in \cite{cox1} addressed this drawback by extending SDE \eqref{sec1:eq:1} to an alternative model often called the mean-reverting square root process which is driven by
\begin{equation}\label{sec1:eq:2}
dx(t)=\alpha_1(\mu_1-x(t))dt+\sigma_1\sqrt{x(t)}dB_1(t).
\end{equation}
The square root diffusion factor avoids possible negative values. Later, Lewis in \cite{lewis} generalised SDE \eqref{sec1:eq:2} to the  mean-reverting-theta process governed by
\begin{equation}\label{sec1:eq:3}
dx(t)=\alpha_1(\mu_1-x(t))dt+\sigma_1x(t)^{\theta}dB_1(t),
\end{equation}
where $\theta\ge 1/2$. The SDE \eqref{sec1:eq:3} has been found a useful tool  for modelling interest rate, asset price and other financial variables. However, by applying $\chi^2$ tests to US treasury bill data, it has been shown that $\theta>1$. For instance, Chan et al. in \cite{Chan} applied the Generalised Moment method to the treasury bill data to estimate $\theta=1.449$. Similarly, with the same data, Nowman in \cite{Nowman} also estimated $\theta=1.361$ using the Gaussian Estimation method. 
\par
There are many literature where several classes of SDE \eqref{sec1:eq:3} with parametric restrictions have been studied. For instance,  Higham and Mao in \cite{highamao} studied strong convergence of Monte Carlo simulations involving SDE \eqref{sec1:eq:3} for $\theta=1/2$. Mao in \cite{maobook} studied strong convergence of EM method for SDE \eqref{sec1:eq:3} when $\theta \in [1/2,1]$. Wu et al. in \cite{wu} established weak convergence of EM method for $\theta > 1$. Dong-Hyun et. al. documented a unique type of SDE \eqref{sec1:eq:3} in \cite{Ahn} which admits closed-form solutions for bond prices and a concave relationship between interest rates and yields. Further discussions relating to SDE \eqref{sec1:eq:3} could also be found in \cite{Yang}, \cite{implicit}, among others. 
\par
The original Black-Scholes model assumes constant volatility for asset price and even for options with different maturities and strikes over a trading period. This assumption makes the Black-Scholes model reproduce flat volatility surface in option pricing. However, in practice, volatility has been observed empirically to change as asset price changes. Essentially, this means that volatility is characterised by a smile or skew surface instead of a flat surface. This characteristic is important for pricing and evaluating complex financial derivatives. As a result, several authors have proposed a variety of volatility models to explain the volatility surface curve adequately. For instance, Dupire (1994) developed the local volatility model in \cite{Dupire} to precisely match the observed smile or skew surface of market volatility data. Subsequently, stochastic volatility models have also been introduced as alternative models for modelling the random nature of volatility through time. One of the most notable stochastic volatility models is the diffusion class of Constant Elasticity of Variance (CEV) model driven by
\begin{equation}\label{sec1:eq:4}
d\varphi(t)=\mu_2\varphi(t)dt+\sigma_2\varphi(t)^{\phi}dB_2(t),
\end{equation}
with initial data $\varphi(0)=\varphi_0$, $\mu_2,\sigma_2>0$, $\phi>1$ and $B_2(t)$ is a scalar Brownian motion. SDE \eqref{sec1:eq:4} is widely used by researchers and market practitioners for modelling volatility and other financial quantities (see, e.g., \cite{CoxJ1,CoxJ2}). In 2012, the authors in \cite{Baduraliya} established the weak convergence result of the Hull and White type model where the instantaneous volatility follows
\begin{equation}\label{sec1:eq:5}
d\varphi(t)=\alpha_2(\mu_2-\varphi(t))dt+\sigma_2\varphi(t)^{\phi}dB_2(t),
\end{equation}
for $\phi>1$. The reader is referred, for example, to \cite{HullWhite,Hagan,Heston} for further coverage of stochastic volatility models in finance.
\par
From the empirical viewpoint, it would be more desirable in modelling context to generalise SDE \eqref{sec1:eq:3} as a highly non-linear SDE of the form
\begin{equation}\label{sec1:eq:6}
dx(t)=\alpha_1(\mu_1-x(t)^{\rho})dt+\sigma_1\sqrt{\varphi(t)}x(t)^{\theta}dB_1(t)
\end{equation}
for asset price dynamics, where $\rho>1$. Here, the variance function $\varphi(t)$ is driven by a highly non-linear type of SDE \eqref{sec1:eq:5} of the form
\begin{equation}\label{sec1:eq:7}
d\varphi(t)=\alpha_2(\mu_2-\varphi(t)^{r})dt+\sigma_2\varphi(t)^{\phi}dB_2(t),
\end{equation}
where $r>1$ and $B_1(t)$ is independent of $B_2(t)$. 
\par
The highly non-linear component of SDE \eqref{sec1:eq:6} makes it well-suited for explaining non-linearity in asset price. On the other hand, the inherent super-linear CEV dynamics may capture extreme non-linearity in market volatility to reproduce volatility surface curve adequately. Obviously, SDE \eqref{sec1:eq:6} is not analytically tractable. The drift and diffusion terms are of super-linear growth. In this case, we recognise the need to develop an implementable numerical method to estimate the exact solution. However, to the best of our knowledge, there exists no relevant literature devoted to the convergent approximation of the system of SDE \eqref{sec1:eq:6} in the strong sense. In this paper, we aim to close this gap by constructing several new numerical tools to study this model from viewpoint of financial applications. 
\par
The rest of the paper is organised as follows: In Section \ref{sect2}, we introduce some useful mathematical notations. In Section \ref{sect3}, we study the existence of a unique positive solution of SDE \eqref{sec1:eq:6} and establish the finite moment of the solution. We construct a new truncated EM method to approximate SDE \eqref{sec1:eq:6} in Section \ref{sect4}. In Section \ref{sect5}, we study numerical properties such as the finite moment and the finite time strong convergence of the numerical solutions. We implement numerical examples to validate the theoretical findings and conclude the paper with a financial application in Section \ref{sect6}.
\section{Preliminaries}\label{sect2}
Throughout this paper unless specified otherwise, we employ the following notation. Let $\{ \Omega,\mathcal{F},\mathbb{P}\}$ be a complete probability space with a filtration $\{ \mathcal{F}_t\}_{t\geq 0}$ satisfying the usual conditions (i.e., it is increasing and right continuous while $\mathcal{F}_0$ contains all $\mathbb{P}$ null sets), and let $\mathbb{E}$ denote the expectation corresponding to $\mathbb{P}$. Let $B_1(t)$ and $B_2(t)$, $t\geq 0$, be scalar Brownian motions defined on the above probability space and are independent of each other. If $x, y$ are real numbers, then we denote $x\vee y$ as the maximum of $x$ and $y$, and $x\wedge y$ as the minimum of $x$ and $y$. Let $\mathbb{R}=(-\infty,\infty)$ and $\mathbb{R}_+=(0,\infty)$. For an empty set $\emptyset$, we set $\inf\emptyset=\infty$. For a set $A$, we denote its indicator function by $1_A$. Moreover, we let $T$ be an arbitrary positive number. Now consider the following scalar dynamics
\begin{align}
dx(t)&=f_1(x(t))dt+\sqrt{\varphi(t)}g_1(x(t))dB_1(t)\label{sec2:eq:1}\\
d\varphi(t)&=f_2(\varphi(t))dt+g_2(\varphi(t))dB_2(t)\label{sec2:eq:2},
\end{align}
as equations of SDEs \eqref{sec1:eq:6} and \eqref{sec1:eq:7}, where $f_1(x)=\alpha_1(\mu_1-x^{\rho})$, $g_1(x)=\sigma_1x^{\theta}$, $f_2(\varphi)=\alpha_2(\mu_2-\varphi^{r})$ and $g_2(\varphi)=\sigma_2\varphi^{\phi}$. Let $H\in C^{2,1}(\mathbb{R}\times \mathbb{R}_+;\mathbb{R})$, where $C^{2,1}(\mathbb{R}\times \mathbb{R}_+;\mathbb{R})$ is the family of all real-valued functions $H(\cdot,\cdot)$ defined on $\mathbb{R}\times \mathbb{R}_+$. Also let $LH:\mathbb{R}\times \mathbb{R}_+\rightarrow \mathbb{R}$ be the It\^{o} diffusion operator such that
\begin{align}
LH(x,\varphi,t)&=H_t(x,t)+H_x(x,t)f_1(x)+\varphi\frac{1}{2}H_{xx}(x,t) g_1(x)^2\label{sec2:eq:3}\\
LH(\varphi,t)&=H_t(\varphi,t)+H_{\varphi}(\varphi,t)f_2(\varphi)+\frac{1}{2}H_{\varphi\varphi}(\varphi,t)g_2(\varphi)^2\label{sec2:eq:4},
\end{align}
where $H_t(x,t)$, $H_t(\varphi,t)$, $H_x(x,t)$ and $H_{\varphi}(\varphi,t)$  are first-order partial derivatives with respect to $t$, $x$ and $\varphi$, and, $H_{xx}(x,t)$ and $H_{\varphi\varphi}(\varphi,t)$ are second-order partial derivatives with respect to $x$ and $\varphi$ respectively. Given the diffusion operator, we can now write the It\^{o} formula as
\begin{align}
dH(x(t),t)&=LH(x(t),\varphi(t),t)dt+\sqrt{\varphi(t)}H_x(x(t),t)g_1(x(t))dB_1(t)\label{sec2:eq:5}\\
dH(\varphi(t),t)&=LH(\varphi(t),t)dt+H_{\varphi}(\varphi(t),t)g_2(\varphi(t))dB_2(t)\label{sec2:eq:6} \quad \text{ a.s. }
\end{align}
The reader may refer to \cite{maobook} for further details about the It\^{o} formula. 
\section{Theoretical properties}\label{sect3}
In this section, we discuss pathwise existence of unique positive solutions and finite moments of the solutions to SDEs \eqref{sec2:eq:1} and \eqref{sec2:eq:2}. The following assumption on the parameters is crucial to obtain the results.
\begin{assumption} \label{sec3:assump:1} 
The parameters of \textup{SDEs}  \eqref{sec2:eq:1} and \eqref{sec2:eq:2} satisfy
\begin{align}
1+ \rho &> 2\theta\label{sec3:eq:1},\\
1+ r&> 2\phi\label{sec3:eq:2}, 
\end{align}
for $ \rho,\theta,\phi,r>1$.
\end{assumption}
\subsection{Existence and uniqueness of solution}
\begin{lemma}\label{sec3:eq:000}
Let equation \eqref{sec3:eq:2} hold. Then there exists a unique global solution $\varphi(t)$ to \textup{SDE} \eqref{sec2:eq:2} on $t\in [0,T]$ for any given initial data $\varphi_0>0$ and $\varphi(t)>0$ a.s.
\end{lemma}
\begin{proof}
Apparently, the coefficients of SDE \eqref{sec2:eq:2} are locally Lipschitz continuous in $\mathbb{R}$. Hence there exists a unique positive maximal local solution $\varphi(t)$ on $t\in [0,\tau_e)$, where $\tau_e$ is the explosion time (e.g., see \cite{maobook,Baduraliya}). Let us extend the domain of SDE \eqref{sec2:eq:2} from $\mathbb{R}_+$ to $\mathbb{R}$ by setting the coefficients to 0 for $\varphi(t)<0$. Then for every sufficiently large integer $n>0$, such that $\varphi(0)\in (1/n,n)$, define the stopping time as 
\begin{equation}\label{sec3:eq:3}
\tau_n=\inf\{t\in [0,\tau_e):\vert\varphi(t)\vert\notin (1/n,n)\}
\end{equation}
and set $\tau_{\infty}=\lim_{n\rightarrow \infty}\tau_n$. To complete the proof, we need to show that $\tau_{\infty}=\infty$ a.s. That is, it is enough to prove that $\mathbb{P}(\tau_n\le T)\rightarrow 0$ as $n\rightarrow \infty$ for any given $T>0$ and hence, $\mathbb{P}(\tau_{\infty}=\infty)=1$. For $\gamma\in(0,1)$, define a $C^2$-function $H:\mathbb{R}_+\rightarrow \mathbb{R}_+$ by
\begin{equation}\label{sec3:eq:4}
H(\varphi)=\varphi^{\gamma}-1-\gamma\log(\varphi).
\end{equation}
Apparently, $H(\varphi)\rightarrow \infty$ as $\varphi\rightarrow 0$ or $\varphi\rightarrow \infty$. By \eqref{sec2:eq:4}, we compute
\begin{align*}
LH(\varphi)&=\gamma(\varphi^{\gamma-1}-1/\varphi)f_2(\varphi)+\frac{1}{2}(\gamma(\gamma-1)\varphi^{\gamma-2}+\gamma\varphi^{-2})g_2(\varphi)^2\\
&=\alpha_2\mu_2\gamma\varphi^{\gamma-1}-\alpha_2\gamma            \varphi^{\gamma+r-1}-\alpha_2\mu_2\gamma\varphi^{-1}+\alpha_2\gamma\varphi^{r-1}\\
&+ \frac{\sigma_2}{2}\gamma(\gamma-1)\varphi^{\gamma+2\phi-2}
           +\frac{\sigma_2}{2}\gamma\varphi^{2\phi-2}.
\end{align*}
So clearly, for $\gamma\in(0,1)$ and by \eqref{sec3:eq:2}, we infer $-\alpha_2\mu_2\gamma\varphi^{-1}$ leads and tends to $-\infty$ for small $\varphi$. Similarly, we infer $-\alpha_2\gamma \varphi^{\gamma+r-1}$ leads and also tends to $-\infty$ for large $\varphi$. So there exists a constant $K_0$
such that
\begin{equation}\label{sec3:eq:5}
LH(\varphi)\leq K_0.
\end{equation} 
By the It\^{o} formula, we have 
\begin{equation}
\mathbb{E}[H(\varphi(T\wedge\tau_n))]\le H(\varphi_0)+K_0T.
\end{equation}
It then follows that,
\begin{equation}
\mathbb{P}(\tau_n\le T)\le \frac{H(\varphi_0)+K_0T}{H(1/n)\wedge H(n)}.
\end{equation}
This implies that $\lim_{n\rightarrow \infty}\mathbb{P}(\tau_n\le T)\rightarrow 0$ as required. 
\end{proof}
\begin{lemma}\label{sec3:eq:00}
Let Assumption \ref{sec3:assump:1} hold. Then for any given initial data $x_0>0$ and $\varphi_0>0$, there exists a unique global solution $x(t)$ to \textup{SDE} \eqref{sec2:eq:1} on $t\in [0,T]$ and $x(t)>0$ a.s. 
\end{lemma}
\begin{proof}
Similarly, we treat SDE \eqref{sec2:eq:1} as an SDE in $\mathbb{R}^2$ by setting its coefficients to 0 whenever $x(t)<0$ or $\varphi(t)<0$. Obviously, the coefficients are locally Lipschitzian. Thus there exists a unique positive maximal local solution $(x(t),\varphi(t))$ on $t\in [0,\tau_e)$, where $\tau_e$ is the explosion time (e.g., see \cite{Baduraliya}). So for any sufficiently large integer $n>0$, define the stopping times 
\begin{align*}
\varsigma_n&=\tau_e\wedge\inf\{t\in [0,\tau_e):\vert x(t)\vert\notin (1/n,n)\},\\
\tau_m&=\tau_e\wedge\inf\{t\in [0,\tau_e):\vert\varphi(t)\vert\notin (1/m,m)\}.
\end{align*}
Now let
\begin{equation}\label{sec3:eq:5}
\varrho_{mn}=\varsigma_n\wedge \tau_m.
\end{equation}
Set $\varsigma_{\infty}=\lim_{n\rightarrow \infty}\varsigma_n$ and $\tau_{\infty}=\lim_{m\rightarrow \infty}\tau_m$. Define a $C^2$-function by
\begin{equation}\label{sec3:eq:6}
H(x)=x^{\gamma}-1-\gamma\log(x).
\end{equation}
for $\gamma\in (0,1)$. Then for $s\in[0,T\wedge \varrho_{mn}]$, we apply \eqref{sec2:eq:3} to compute 
\begin{align*}
&LH(x(s),\varphi(s))\\
&=\gamma(x(s)^{\gamma-1}-1/x(s))f_1(x(s))+\frac{1}{2}(\gamma(\gamma-1)x(s)^{\gamma-2}+\gamma x(s)^{-2})\varphi(s) g_1(x(s))^2\\
&\le \alpha_1\mu_1\gamma x(s)^{\gamma-1}-\alpha_1\gamma x(s)^{\gamma+\rho-1}-\alpha_1\mu_1\gamma x(s)^{-1}+\alpha_1\gamma x(s)^{\rho-1}\\
&+ \frac{\sigma_1}{2}\gamma(\gamma-1)x(s)^{\gamma+2\theta-2}
 +\frac{m\sigma_1}{2}\gamma x(s)^{2\theta-2},
\end{align*}
Moreover, we can derive that
\begin{align*}
\mathbb{E}[H(x(T\wedge \varrho_{mn}))]
&=\mathbb{E}[H(x(T\wedge \varsigma_n\wedge \tau_m))] \ge \mathbb{E}[H(x(\varsigma_n))1_{(\varsigma_n\le T\wedge\tau_m )}]\\
&\ge [H(1/n)\wedge H(n)]\mathbb{P}(\varsigma_n\le T\wedge\tau_m).
\end{align*}
By the It\^{o} formula, we derive that
\begin{align*}
\mathbb{E}[H(x(T\wedge \varrho_{mn}))]\le H(x_0)+ \mathbb{E}\int_0^{T\wedge \varrho_{mn}}LH(x(s),\varphi(s))ds.
\end{align*}
This implies 
\begin{align*}
[H(1/n)\wedge H(n)]\mathbb{P}(\varsigma_n\le T\wedge\tau_m)\le H(x_0)+ \mathbb{E}\int_0^{T\wedge \varrho_{mn}}LH(x(s),\varphi(s))ds.
\end{align*}
Meanwhile, for $\gamma\in(0,1)$ and by \eqref{sec3:eq:2}, we can find a constant $K_1$ such that
\begin{equation}
[H(1/n)\wedge H(n)]\mathbb{P}(\varsigma_n\le T\wedge\tau_m)\le H(x_0)+K_1T.
\end{equation}
This means we have
\begin{equation}
\mathbb{P}(\varsigma_n\le T\wedge \tau_m)\le\frac{H(x_0)+K_1T}{H(1/n)\wedge H(n)}.
\end{equation}
So, by letting $n\rightarrow \infty$, we obtain $\mathbb{P}(\varsigma_n\le T\wedge \tau_m)\rightarrow 0$. By setting $m\rightarrow \infty$ and using Lemma \ref{sec3:eq:000}, we have $\mathbb{P}(\varsigma_{\infty}\le T)=0$. This implies $\mathbb{P}(\varsigma_{\infty}>T)=1$.
\end{proof}
\subsection{Finite moments}
In the sequel, we show that the moments of SDEs \eqref{sec2:eq:1} and \eqref{sec2:eq:2} are finite.
\begin{lemma}\label{sec3:eq:L1}
Let equation \eqref{sec3:eq:2} hold. Then for any $p\geq 2$, the solution $\varphi(t)$ to \textup{SDE}  \eqref{sec2:eq:2} satisfies
\begin{equation}\label{sec3:eq:8}
\sup_{0\leq t<\infty}(\mathbb{E}|\varphi(t)|^p)\leq c_1,
\end{equation}
where $c_1$ is a constant.
\end{lemma}
See \cite{emma} for the proof.
\begin{lemma}\label{sec3:eq:L2}
Let Assumption \ref{sec3:assump:1} hold. Then for any $p\geq 2$, the solution $x(t)$ to \textup{SDE}  \eqref{sec2:eq:2} obeys
\begin{equation}\label{sec3:eq:9}
\sup_{0\leq t<\infty}\big(\mathbb{E}|x(t)|^p1_{(t\le \tau^*_m )}\big)\leq c_2,
\end{equation}
where for any sufficiently large integer $m>0$,
\begin{align*}
\tau_m^*&=\inf\{t\ge 0: \varphi(t)\notin (1/m,m)\}
\end{align*}
and $c_2$ is a constant dependent on $m$.
\end{lemma}
\begin{proof}
For any sufficiently large integer $n>0$, define the stopping times 
\begin{align*}
\varsigma^*_n&=\inf\{t\ge 0: x(t)\notin (1/n,n)\}.
\end{align*}
Then set $\varrho^*_{mn}=\varsigma^*_n\wedge \tau_n^*$. For $s\in[0,t\wedge \varrho^*_{mn}]$, we apply \eqref{sec2:eq:4} to $H(x,t)=e^tx^p$ to compute
\begin{align*}
&LH(x(s),\varphi(s))\\
&=e^sx(s)^p+pe^sx(s)^{p-1}f_1(x(s))+\frac{1}{2}p(p-1)e^sx(s)^{p-2}\varphi(s)g_1(x(s))^2\\
&=e^s\big(x(s)^p+\alpha_1\mu_1 px(s)^{p-1}-\alpha_1 px(s)^{\rho+p-1}+\frac{\sigma_1}{2}p(p-1)\varphi(s)x(s)^{2\theta+p-2}\big)\\
&\le e^s\big(x(s)^p+\alpha_1\mu_1 px(s)^{p-1}-\alpha_1 px(s)^{\rho+p-1}+\frac{m\sigma_1}{2}p(p-1)x(s)^{2\theta+p-2}\big),
\end{align*}\\
By the It\^{o} formula, we get
\begin{align*}
\mathbb{E}[e^{t\wedge \varrho^*_{mn}}|x(t\wedge \varrho^*_{mn})|^p]&\le x_0^p+\mathbb{E}\int_0^{t\wedge \varrho^*_{mn}}LH(x(s),\varphi(s))ds.
\end{align*}
Noting that
\begin{align*}
\mathbb{E}[e^{t\wedge \varrho^*_{mn}}|x(t\wedge \varrho^*_{mn})|^p]&=\mathbb{E}[e^{t\wedge\varsigma^*_n\wedge \tau_n^*}|x(t\wedge \varsigma^*_n\wedge \tau_n^*)|^p]\\
&\ge \mathbb{E}[e^{t\wedge \varsigma^*_n}|x(t\wedge\varsigma^*_n)|^p1_{(t\wedge\varsigma^*_n\le \tau^*_m )}],
\end{align*}
we obtain
\begin{align*}
\mathbb{E}[e^{t\wedge \varsigma^*_n}|x(t\wedge\varsigma^*_n)|^p1_{(t\wedge\varsigma^*_n\le \tau^*_m )}]&\le x_0^p+\mathbb{E}\int_0^{t\wedge \varrho^*_{mn}}LH(x(s),\varphi(s))ds.
\end{align*}
So, by Assumption \ref{sec3:assump:1}, we can find a constant $K_3$ such that 
\begin{equation*}
\mathbb{E}[e^{t\wedge \varsigma^*_n}|x(t\wedge\varsigma^*_n)|^p1_{(t\wedge\varsigma^*_n\le \tau^*_m )}]\le x_0^p+e^tK_3
\end{equation*}
By letting $n\rightarrow \infty $, we can apply the Fatou lemma to have
\begin{equation*}
\mathbb{E}|x(t)|^p1_{(t\le \tau^*_m )}\le x_0^pe^{-t}+K_3,
\end{equation*}
and consequently,
\begin{equation*}
\sup_{0\leq t<\infty}\big(\mathbb{E}|x(t)|^p1_{(t\le \tau^*_m )}\big)\leq c_2
\end{equation*}
as the desired assertion. The proof is thus complete.
\end{proof}
\section{Numerical method}\label{sect4}
In this section, we construct the truncated EM method to approximate SDEs \eqref{sec2:eq:1} and \eqref{sec2:eq:2}. But before then, we need to introduce the following lemmas which are needed to perform the convergence analysis (see \cite{mao3}).
\begin{lemma}\label{sec4:eq:L1}
For any $R>0$, there exist positive constants $K_R$ and $L_R$ such that the coefficients of \textup{SDE} \eqref{sec2:eq:1} and \textup{SDE} \eqref{sec2:eq:2} satisfy
\begin{align}
|f_1(x)-f_1(\bar{x})|\vee |g_1(x)-g_1(\bar{x})|&\le K_R|x-\bar{x}|\label{sec4:eq:1},\\
|f_2(\varphi)-f_2(\bar{\varphi})|\vee |g_2(\varphi)-g_2(\bar{\varphi})|&\le L_R|\varphi-\bar{\varphi}|\label{sec4:eq:2}
\end{align}
for all $\varphi,\bar{\varphi} \in \mathbb{R}$ and $x,\bar{x}\in \mathbb{R}^2$ with $|x|\vee|\bar{x}|\vee|\varphi|\vee|\bar{\varphi}|\le R$.
\end{lemma}
\begin{lemma}\label{sec4:eq:L2}
Let Assumption \ref{sec3:assump:1} hold. Then for any $p\geq 2$, there exist $K_4=K(p)>0$ and $K_5=K(p)>0$ such that the coefficients terms of \textup{SDE}  \eqref{sec2:eq:1} and \eqref{sec2:eq:2} fulfil
\begin{align}
xf_1(x)+\frac{p-1}{2}|\sqrt{\varphi}g_1(x)|^2&\le K_4(1+\varphi|x|^2)\label{sec4:eq:3}\\
\varphi f_2(\varphi)+\frac{p-1}{2}|g_2(\varphi)|^2&\le K_5(1+|\varphi|^2)\label{sec4:eq:4}
\end{align}
$\forall\varphi\in \mathbb{R}_+$, $\forall x\in \mathbb{R}^2_+$. See \cite{mao3} for the proof.
\end{lemma}
\subsection{Numerical schemes}
To begin with, let us extend the domain of SDE \eqref{sec2:eq:2} from $\mathbb{R}_+$ to $\mathbb{R}$ and SDE \eqref{sec2:eq:1} from $\mathbb{R}^2_+$ to $\mathbb{R}^2$. We should mention that these extensions do not affect the positivity of the solutions and the local Lipschitz conditions. We define the truncated scheme by first choosing a strictly increasing continuous function $\nu:\mathbb{R}_+\rightarrow \mathbb{R}_+$ such that $\nu(r)\rightarrow\infty$ as $r\rightarrow \infty$ and 
\begin{equation}\label{sec4:eq:5}
\sup_{\vert x\vert\vee\vert\varphi\vert \le r}\Big(|f_1(x)|\vee |f_2(\varphi)|\vee g_1(x)\vee g_2(\varphi)\Big)\leq \nu(r), \quad \forall r\ge 0.
\end{equation}
Denote by $\nu^{-1}$ the inverse function of $\nu$ and we see that $\nu^{-1}$ is strictly increasing continuous function from $[\nu(0),\infty)$ to $\mathbb{R}_+$. We also choose a number $\Delta^*\in (0,1]$ and a strictly decreasing function 
$h:(0,\Delta^*]\rightarrow (0,\infty)$ such that
\begin{equation}\label{sec4:eq:6}
\quad h(\Delta^*)\ge \nu(1), \lim_{\Delta \rightarrow 0}h(\Delta)=\infty \text{ and } \Delta^{1/4}h(\Delta)\leq 1 \quad \forall \Delta \in (0,1].
\end{equation}
For any given step size $\Delta \in (0,1)$, we define the truncated functions by
\begin{equation*}
f_1^{\Delta}(x)=
\begin{cases}
  f_1\Big(x\wedge \nu^{-1}(h(\Delta))\Big), & \mbox{if $x\geq 0$ }\\
  \alpha_1\mu_1,                             & \mbox{if $x<0$},
\end{cases}
\end{equation*}
\begin{equation*}
g_1^{\Delta}(x)=
\begin{cases}
  g_1\Big(x\wedge \nu^{-1}(h(\Delta))\Big), & \mbox{if $x\geq 0$ }\\
  0,                             & \mbox{if $x<0$},
\end{cases}
\end{equation*}
\begin{equation*}
f_2^{\Delta}(\varphi)=
\begin{cases}
f_2\Big(\varphi\wedge \nu^{-1}(h(\Delta))\Big), & \mbox{if $\varphi \geq 0$ }\\
   \alpha_2\mu_2,    & \mbox{if $\varphi<0$}.
\end{cases}
\end{equation*}
and
\begin{equation*}
g_2^{\Delta}(\varphi)=
\begin{cases}
g_2\Big(\varphi\wedge \nu^{-1}(h(\Delta))\Big), & \mbox{if $\varphi \geq 0$ }\\
  0,    & \mbox{if $\varphi<0$},
\end{cases}
\end{equation*}
for $ \varphi\in \mathbb{R}$ and $x\in \mathbb{R}^2$ .
Clearly, we observe
\begin{align}\label{sec4:eq:6*}
|f_1^{\Delta}(x)|\vee |f_2^{\Delta}(\varphi)|\vee g_1^{\Delta}(x)\vee g_2^{\Delta}(\varphi)\le \nu( \nu^{-1}(h(\Delta)))=h(\Delta)
\end{align}
for $ \varphi\in \mathbb{R}$, $x\in \mathbb{R}^2$. The truncated functions $f_1^{\Delta}$ and $g_1^{\Delta}(x)$, and $f_2^{\Delta}$ and $g_2^{\Delta}$ maintain \eqref{sec4:eq:2} and \eqref{sec4:eq:3} respectively as shown in the following lemma.
\begin{lemma}\label{sec4:eq:L3}
Let Assumption \ref{sec3:assump:1} hold. Then, for all $\Delta \in (0,\Delta^*)$ and $p\geq 2$, the truncated functions satisfy
\begin{align}
xf_1^{\Delta}(x)+\frac{p-1}{2}|\sqrt{\varphi}g_1^{\Delta}(x)|^2&\le K_6(1+\varphi|x|^2)\label{sec4:eq:7},\\
\varphi f_2^{\Delta}(\varphi)+\frac{p-1}{2}|g_2^{\Delta}(\varphi)|^2&\le K_7(1+|\varphi|^2)\label{sec4:eq:8}
\end{align}
 $\forall\varphi\in \mathbb{R}_+$, $\forall x\in \mathbb{R}^2_+$, where $K_6$ and $K_7$ are independent of $\Delta$.
\end{lemma}
\begin{proof}
See \cite{mao3} for the proof of \eqref{sec4:eq:8}.  To prove \eqref{sec4:eq:7}, fix any $\Delta\in (0,\Delta^*]$. Then for $\varphi\in \mathbb{R}$ and $x\in \mathbb{R}^2$ with $\vert x\vert\vee\vert\varphi\vert\le \nu^{-1}(h(\Delta))$, by \eqref{sec4:eq:4}, we obtain
\begin{align*}
xf_1^{\Delta}(x)+\frac{p-1}{2}|\sqrt{\varphi}g_1^{\Delta}(x)|^2&\le xf_1(x)+\frac{p-1}{2}|\sqrt{\varphi}g_1(x)|^2\\
&\le K_4(1+\varphi|x|^2)
\end{align*}
as required. For $\varphi\in \mathbb{R}$ and $x\in \mathbb{R}^2$ with $\vert x\vert\vee\vert\varphi\vert>\nu^{-1}(h(\Delta))$, we get
\begin{align*}
xf_1^{\Delta}(x)+\frac{p-1}{2}|\sqrt{\varphi}g_1^{\Delta}(x)|^2
&\le xf_1(\nu^{-1}(h(\Delta)))+ \frac{p-1}{2}|\sqrt{\varphi} g_1(\nu^{-1}(h(\Delta)))|^2\\
&\le \nu^{-1}(h(\Delta))f_1(\nu^{-1}(h(\Delta)))\\
&+ \frac{p-1}{2}|\sqrt{\nu^{-1}(h(\Delta))}g_1(\nu^{-1}(h(\Delta))|^2\\
&+\Big(\frac{x}{\nu^{-1}(h(\Delta))}-1 \Big)\nu^{-1}(h(\Delta))f_1(\nu^{-1}(h(\Delta)))\\
&\le K_4(1+\nu^{-1}(h(\Delta))[\nu^{-1}(h(\Delta))]^2)\\
&+\Big(\frac{x}{\nu^{-1}(h(\Delta))}-1\Big)\nu^{-1}(h(\Delta))f_1(\nu^{-1}(h(\Delta))).
\end{align*}
Again, we observe from \eqref{sec4:eq:3} that $xf_1(x)\le K_4(1+\varphi|x|^2)$ for any $\varphi\in\mathbb{R}$ and $x\in\mathbb{R}^2$, we obtain
\begin{align*}
xf_1^{\Delta}(x)+\frac{p-1}{2}|\sqrt{\varphi}g_1^{\Delta}(x)|^2
&\le K_4(1+\nu^{-1}(h(\Delta))[\nu^{-1}(h(\Delta))]^2)\\
&+\Big(\frac{x}{\nu^{-1}(h(\Delta))}-1\Big)K_4(1+\nu^{-1}(h(\Delta))[\nu^{-1}(h(\Delta))]^2)\\
&\le \frac{x}{\nu^{-1}(h(\Delta))}K_4(1+\nu^{-1}(h(\Delta))[\nu^{-1}(h(\Delta))]^2)\\
&\le x\cdot K_4(1+\nu^{-1}(h(\Delta))[\nu^{-1}(h(\Delta))])\\
&\le x\cdot K_5(1+\varphi\cdot x)\leq K_7(1+\varphi|x|^2),
\end{align*}
where $K_7=2K_4$ as the required assertion in \eqref{sec4:eq:7}. We should mention that using these proofs, we could similarly establish the case when $\varphi\in \mathbb{R}$ and $x\in \mathbb{R}^2$ with $\vert x\vert>\nu^{-1}(h(\Delta))$ and $\vert\varphi\vert\le \nu^{-1}(h(\Delta))$ and the case when $\varphi\in \mathbb{R}$ and $x\in \mathbb{R}^2$ with $\vert x\vert\le\nu^{-1}(h(\Delta))$ and $\vert\varphi\vert>\nu^{-1}(h(\Delta))$.
\end{proof}
Let us now form the discrete-time truncated EM solutions $Y_{\Delta}(t_k)\approx \varphi(t_k)$ and $X_{\Delta}(t_k)\approx x(t_k)$ to SDEs \eqref{sec2:eq:1} and \eqref{sec2:eq:2} for $t_k=k\Delta$ respectively, by setting $Y_{\Delta}(0)=\varphi_0$, $X_{\Delta}(0)=x_0$ and computing
\begin{align}
Y_{\Delta}(t_{k+1})&=Y_{\Delta}(t_k)+f_2^{\Delta}(Y_{\Delta}(t_k))\Delta+g_2^{\Delta}(Y_{\Delta}(t_{k}))\Delta B_{2k}\label{sec4:eq:9}\\
X_{\Delta}(t_{k+1})&=X_{\Delta}(t_k)+f_1^{\Delta}(X_{\Delta}(t_k))\Delta+\sqrt{|Y_{\Delta}(t_k)|}g_1^{\Delta}(X_{\Delta}(t_{k}))\Delta B_{1k}\label{sec4:eq:10}
\end{align}
for $k=0,1,2,\cdots,$ where $\Delta=t_{k+1}-t_k$, $\Delta B_{1k}=(B_1(t_{k+1})-B_1(t_k))$ and $\Delta B_{2k}=(B_2(t_{k+1})-B_2(t_k))$. Let us now form corresponding versions of the continuous-time truncated EM solutions. The first versions are defined by
\begin{align}
\bar{\varphi}_{\Delta}(t)&=\sum_{k=0}^{\infty}Y_{\Delta}(t_k)1_{[t_k,t_{k+1})}(t)\label{sec4:eq:10}\\
\bar{x}_{\Delta}(t)&=\sum_{k=0}^{\infty}X_{\Delta}(t_k)1_{[t_k,t_{k+1})}(t)\label{sec4:eq:11}.
\end{align}
on $t\ge0$. These are the continuous-time step processes. The other versions are the continuous-time continuous processes defined on $t\ge 0$ by
\begin{align}
\varphi_{\Delta}(t)=\varphi(0)+\int_0^t f_2^{\Delta}(\bar{\varphi}_{\Delta}(s))ds+\int_0^t g_2^{\Delta}(\bar{\varphi}_{\Delta}(s))dB_2(s)\label{sec4:eq:12}\\
x_{\Delta}(t)=x(0)+\int_0^t f_1^{\Delta}(\bar{x}_{\Delta}(s))ds+\int_0^t \sqrt{|\bar{\varphi}(s)|} g_1^{\Delta}(\bar{x}_{\Delta}(s))dB_1(s)\label{sec4:eq:13}.
\end{align}
Obviously $\varphi_{\Delta}(t)$ and $x_{\Delta}(t)$ are It\^{o} processes on $t\ge 0$ respectively satisfying It\^{o} differentials
\begin{align*}
d\varphi_{\Delta}(t)= f_2^{\Delta}(\bar{\varphi}_{\Delta}(t))dt+g_2^{\Delta}(\bar{\varphi}_{\Delta}(t))dB_2(t)\\
dx_{\Delta}(t)= f_1^{\Delta}(\bar{x}_{\Delta}(t))dt+\sqrt{|\bar{\varphi}_{\Delta}(t)|}g_1^{\Delta}(\bar{x}_{\Delta}(t))dB_1(t).
\end{align*}
For all $k\ge 0$, we clearly observe that $\varphi_{\Delta}(t_{k})=\bar{\varphi}_{\Delta}(t_k)=Y_{\Delta}(t_k)$ and $x_{\Delta}(t_{k})=\bar{x}_{\Delta}(t_k)=X_{\Delta}(t_k)$.
\section{Numerical properties}\label{sect5}
In this section, we establish the moment bounds and finite time strong convergence results for the truncated EM solutions. 
\subsection{Finite moments}
In the sequel, let us recall the following useful lemmas. The proofs of these lemmas are in \cite{mao3} and therefore omitted.
\begin{lemma}\label{sec5:eq:L1}
Let equation \eqref{sec3:eq:2} hold. Then for any $p\ge 2$, the solution of \eqref{sec4:eq:12} satisfies
\begin{equation}\label{sec5:eq:1}
\sup_{0<\Delta\le\Delta^*}\sup_{0\le t\le T}\big(\mathbb{E}\vert \varphi_{\Delta}(t)\vert^p\big)\le c_4, 
\end{equation}
$\forall T\ge 0$ where $c_4:=c_4(\varphi_0,p,T,K_7)$ may change between occurrences. 
\end{lemma}
It is important to note that \eqref{sec5:eq:1} also holds for $\bar{\varphi}_{\Delta}(t)$ because $\varphi_{\Delta}(t_{k})$ and $\bar{\varphi}_{\Delta}(t_k)$ coincide at discrete time $t_k$ for all $k\ge 0$.
\begin{lemma}\label{sec5:eq:L1**}
For any $\Delta\in(0,\Delta^*]$ and $\forall t\ge0$, we have 
\begin{equation}\label{sec5:eq:2}
\mathbb{E}\vert \varphi_{\Delta}(t)-\bar{\varphi}_{\Delta}(t)\vert^p \le c_p\Delta^{p/2} (h(\Delta))^p
\end{equation}
and consequently, 
\begin{equation}\label{sec5:eq:3}
\lim_{\Delta\rightarrow 0}\mathbb{E}\vert \varphi_{\Delta}(t)-\bar{\varphi}_{\Delta}(t)\vert^p=0,
\end{equation}
where $c_p$ is a positive constant which depends only on $p$.
\end{lemma}
In addition to the above lemmas, we also need the following lemmas.
\begin{lemma}\label{sec5:eq:L1***}
For any $\Delta\in(0,\Delta^*]$ and $\forall t\ge0$, we have 
\begin{equation}\label{sec5:eq:4}
\mathbb{E}\vert x_{\Delta}(t)-\bar{x}_{\Delta}(t)\vert^p \le C_p\Delta^{p/2} (h(\Delta))^p
\end{equation}
and consequently, 
\begin{equation}\label{sec5:eq:5}
\lim_{\Delta\rightarrow 0}\mathbb{E}\vert x_{\Delta}(t)-\bar{x}_{\Delta}(t)\vert^p=0,
\end{equation}
where $C_p$ is a positive constant which depends only on $p$.
\end{lemma}
\begin{proof}
Fix any $\Delta\in (0,\Delta^*]$ and $t\ge 0$. Then there is a unique integer $k\ge 0$ such that $t_k\le t\le t_{k+1}$. By elementary inequality, we derive
\begin{align*}
&\mathbb{E}\vert x_{\Delta}(t)-\bar{x}_{\Delta}(t)\vert^p=\mathbb{E}\vert x_{\Delta}(t)-\bar{x}_{\Delta}(t_k)\vert^p\\
&\le c(p)\Big(\mathbb{E}\big\vert\int_{t_k}^tf_1^{\Delta}(\bar{x}_{\Delta}(s))ds\big\vert^p+\mathbb{E}\big\vert\int_{t_k}^t\sqrt{|\bar{\varphi}(s)|}g_1^{\Delta}(\bar{x}_{\Delta}(s))dB(s)\big\vert^p\Big)\\
&\le c(p)\Big(\Delta^{p-1}\mathbb{E}\int_{t_k}^t\vert f_1^{\Delta}(\bar{x}_{\Delta}(s))\vert^p ds+\Delta^{(p-2)/2}\mathbb{E}\int_{t_k}^t\vert \sqrt{|\bar{\varphi}(s)|}g_1^{\Delta}(\bar{x}_{\Delta}(s))\vert^p ds\Big).
\end{align*}
So by Lemma \eqref{sec5:eq:L1} and \eqref{sec4:eq:6*}, we have 
\begin{align*}
\mathbb{E}\vert x_{\Delta}(t)-\bar{x}_{\Delta}(t)\vert^p
&\le c(p)\Big(\Delta^{p-1}(h(\Delta))^p\Delta+|c_4|^{p/2}\Delta^{(p-2)/2}(h(\Delta))^p\Delta\Big)\\
&\le c(p)\Big(\Delta^{p}(h(\Delta))^p+|c_4|^{p/2}\Delta^{p/2}(h(\Delta))^p\Big)\\
&\le C_p\Delta^{p/2}(h(\Delta))^p,
\end{align*}
where $C_p=c(p)(1\vee |c_4|^{p/2})$. Nothing that $\Delta^{p/2}(h(\Delta))^p\le \Delta^{p/4}$ from \eqref{sec4:eq:6}, we obtain \eqref{sec5:eq:6} from \eqref{sec5:eq:5} by letting $\Delta\rightarrow 0$.
\end{proof}
\begin{lemma}\label{sec5:eq:L2}
Let Assumption \ref{sec3:assump:1}  hold. Then for any $p\ge 2$, the truncated EM solution of \eqref{sec5:eq:1} satisfies
\begin{equation}\label{sec5:eq:6}
\sup_{0\leq t< \infty}\big(\mathbb{E}|x_{\Delta}(t)|^p1_{(t\le \hbar^*_m)}\big)\le c_5,
\end{equation}
where for any sufficiently large integer $m>0$,
\begin{align*}
\hbar_m&=\inf\{t\ge 0: \bar{\varphi}(t)\notin (1/m,m)\}
\end{align*}
and $c_5:=c_5(x_0,\varphi_0,p,T,K_6,m)$ may change value between occurrences.
\end{lemma}
\begin{proof}
Fix any $\Delta \in (0,\Delta^*)$ and for every sufficiently large integer $n>0$, define 
\begin{align*}
\hbar^*_n&=\inf\{t\ge 0: x_{\Delta}(t)\notin (1/n,n)\}.
\end{align*}
Now set $\eth_{mn}=\hbar_m\wedge \hbar^*_n$. By the It\^{o} formula, we derive from \eqref{sec4:eq:13} that
\begin{align*}
&\mathbb{E}|x_{\Delta}(t\wedge\eth_{mn})|^p-|x_0|^p\\
&\le \mathbb{E}\int_{0}^{t\wedge\hbar_{mn}}p|x_{\Delta}(s)|^{p-2}\Big(x_{\Delta}(s)f_1^{\Delta}(\bar{x}_{\Delta}(s))
+ \frac{p-1}{2}|\sqrt{|\bar{\varphi}(s)|}g_1^{\Delta}(\bar{x}_{\Delta}(s))|^2\Big)ds\\
&=\mathcal{J}_1+\mathcal{J}_2,
\end{align*}
where
\begin{align*}
\mathcal{J}_1&=\mathbb{E}\int_{0}^{t\wedge\eth_{mn}}p|x_{\Delta}(s)|^{p-2}\Big(\bar{x}_{\Delta}(s)f_1^{\Delta}(\bar{x}_{\Delta}(s))+ \frac{p-1}{2}|\sqrt{|\bar{\varphi}(s)|}g_1^{\Delta}(\bar{x}_{\Delta}(s))|^2\Big)ds\\
\mathcal{J}_2&=\mathbb{E}\int_{0}^{t\wedge\eth_{mn}}p|x_{\Delta}(s)|^{p-2}(x_{\Delta}(s)-\bar{x}_{\Delta}(s))f_1^{\Delta}(\bar{x}_{\Delta}(s))ds.
\end{align*}
By the Young inequality, we have 
\begin{align*}
\mathcal{J}_1&=K_6\mathbb{E}\int_{0}^{t\wedge\eth_{mn}}|x_{\Delta}(s)|^{p-2}(1+\bar{\varphi}_{\Delta}(s)|\bar{x}_{\Delta}(s)|^2)ds\\
&\le K_6\mathbb{E}\int_{0}^{t\wedge\eth_{mn}}p\Big(|x_{\Delta}(s)|^{(p-2)\frac{p}{p-2}}\Big)^{\frac{p-2}{p}}\Big((1+|\bar{\varphi}_{\Delta}(s)|^{\frac{p}{2}}|\bar{x}_{\Delta}(s)|)^p\Big)^{\frac{2}{p}}ds\\
&\le K_6\mathbb{E}\int_{0}^{t\wedge\eth_{mn}}\Big((p-2)|x_{\Delta}(s)|^{p}+2(1+|\bar{\varphi}_{\Delta}(s)|^{\frac{p}{2}}|\bar{x}_{\Delta}(s)|^{p})\Big)ds\\
&\le K_6\mathbb{E}\int_{0}^{t\wedge\eth_{mn}}\Big((p-2)|x_{\Delta}(s)|^{p}+2+2m^{\frac{p}{2}}|\bar{x}_{\Delta}(s)|^{p})\Big)ds\\
&\leq r_1\mathbb{E}\int_{0}^{t\wedge\eth_{mn}}(|x_{\Delta}(s)|^{p}+\bar{x}_{\Delta}(s)|^{p})ds,
\end{align*}
where $r_1=K_6[2T+((p-2)\vee 2m^{\frac{p}{2}})]$. Also, by Lemma \ref{sec5:eq:L1}, we have
\begin{align*}
\mathcal{J}_2 &\le \mathbb{E}\int_{0}^{t\wedge\eth_{mn}}p|x_{\Delta}(s)|^{p-2}(x_{\Delta}(s)-\bar{x}_{\Delta}(s))f_1^{\Delta}(\bar{x}_{\Delta}(s))ds\\
   &\leq p\mathbb{E}\Big(\int_{0}^{t\wedge\eth_{mn}} |x_{\Delta}(s)|^{(p-2)\times \frac{p}{(p-2)}}ds\Big)^{\frac{(p-2)}{p}}\Big(\int_{0}^{t\wedge\eth_{mn}}(x_{\Delta}(s)-\bar{x}_{\Delta}(s))^{1\times \frac{p}{2}}f_1^{\Delta}(\bar{x}_{\Delta}(s))^{1\times \frac{p}{2}}ds\Big)^{\frac{2}{p}}\\
   &\leq p\mathbb{E}\Big(\int_{0}^{t\wedge\eth_{mn}} |x_{\Delta}(s)|^{p}ds\Big)^{\frac{(p-2)}{p}}\Big(\int_{0}^{t\wedge\eth_{mn}}(x_{\Delta}(s)-\bar{x}_{\Delta}(s))^{\frac{p}{2}}(f_1^{\Delta}(\bar{x}_{\Delta}(s))^{\frac{p}{2}}ds\Big)^{\frac{2}{p}}.\\
 &\leq (p-2)\mathbb{E}\int_{0}^{t\wedge\eth_{mn}}|x_{\Delta}(s)|^{p}ds+2\int_{0}^{T}\Big(\mathbb{E}|x_{\Delta}(s)-\bar{x}_{\Delta}(s)|f_1^{\Delta}(\bar{x}_{\Delta}(s))\Big)^{\frac{p}{2}}ds\\
&\le (p-2)\mathbb{E}\int_{0}^{t\wedge\eth_{mn}}|x_{\Delta}(s)|^pds+2c_p^{1/2}T\Delta^{p/4}(h(\Delta))^{p}.
\end{align*}
Noting from \eqref{sec4:eq:6} that $[\Delta^{1/4}(h(\Delta))]^{p}\le 1 $, we have 
\begin{align*}
\mathcal{J}_2 &\le r_2\mathbb{E}\int_{0}^{t\wedge\eth_{mn}}|x_{\Delta}(s)|^pds,
\end{align*}
where $r_2=(2c_p^{1/2}T)\vee (p-2) $. We now combine $\mathcal{J}_1$ and $\mathcal{J}_2$ to get
\begin{align*}
\mathbb{E}|x_{\Delta}(t\wedge\eth_{mn})|^p&\le |x_0|^p+\mathbb{E}\int_{0}^{t\wedge\eth_{mn}}(r_1|x_{\Delta}(s)|^p+(r_1+r_2)|\bar{x}_{\Delta}(s)|^{p})ds\\
&\le |x_0|^p+(2r_1+r_2)\int_{0}^{t}\sup_{0\leq t \leq s}\Big(\mathbb{E}|x_{\Delta}(t\wedge\eth_{mn})|^p \Big)ds.
\end{align*}
The Gronwall inequality yields
\begin{equation*}
  \sup_{0\leq t< \infty}(\mathbb{E}|x_{\Delta}(t\wedge\eth_{mn})|^p)\leq c_5
\end{equation*}
 where $c_5=|x_0|^pe^{(2r_1+r_2)}$ is independent of $\Delta$. Noting that
\begin{align*}
\sup_{0\leq t< \infty}(\mathbb{E}|x_{\Delta}(t\wedge\eth_{mn})|^p)&\ge \sup_{0\leq t< \infty}(\mathbb{E}|x_{\Delta}(t\wedge \hbar^*_n)|^p1_{(t\wedge\hbar_n\le \hbar^*_m)}),
\end{align*}
we can set $n\rightarrow \infty$ to obtain
\begin{equation*}
\sup_{0\leq t< \infty}\big(\mathbb{E}|x_{\Delta}(t)|^p1_{(t\le \hbar^*_m)}\big)\le c_5
\end{equation*}
as the desired result. The proof is now complete.
\end{proof}
\subsection{Strong convergence}
Before we establish the main result in this section, we need the following lemmas. The proofs of these lemmas could be found in \cite{emma}.
\begin{lemma}\label{eq:L3}
Let equation \eqref{sec3:eq:2} hold and $T>0$ be fixed. Then for any $\epsilon\in (0,1)$, there exists a pair of positive constants $n=n(\epsilon)$ and $\Delta^1=\Delta^1(\epsilon)$ such that for each $\Delta\in (0,\Delta^1]$, we have
\begin{equation}\label{sec5:eq:7}
  \mathbb{P}(\vartheta_n\le T)\le \epsilon,
\end{equation}
where 
\begin{equation}\label{sec5:eq:8}
\vartheta_n=\vartheta(\Delta,n)=\inf\{t\in [0,T]:\varphi_{\Delta}(t)\notin (1/n,n)\}.
\end{equation}
is a stopping time.
\end{lemma}
\begin{lemma}\label{eq:l8}
Let equation \eqref{sec3:eq:2} hold. Then for any $p\geq 2$, $T> 0$, we have
\begin{equation}\label{sec5:eq:9}
\mathbb{E}\Big( \sup_{0\leq t \leq T}|\varphi_{\Delta}(t \wedge \upsilon_n)-\varphi(t \wedge \upsilon_n)|^p  \Big)\le \mathcal{K}_1\Delta^{p/4} 
\end{equation}
for any sufficiently large $n$ and any $\Delta\in (0,\Delta^*]$, where $\mathcal{K}_1$ is a constant independent of $\Delta$ and $\upsilon_n$ is a stopping time. Consequently, we have
\begin{equation}\label{sec5:eq:10}
\lim_{\Delta\rightarrow 0}\mathbb{E}\Big( \sup_{0\leq t \leq T}|\varphi_{\Delta}(t \wedge \upsilon_n)-\varphi(t \wedge\upsilon_n)|^p \Big)=0.
\end{equation}
\end{lemma}
Let us proceed to establish the following useful lemmas.
\begin{lemma}\label{eq:L3}
Let Assumption \ref{sec3:assump:1} hold and $T>0$ be fixed. Define a stopping time by 
\begin{equation}\label{eq:7*}
\vartheta^*_n=\vartheta^*(\Delta,n)=\inf\{t\in [0,T]:x_{\Delta}(t)\notin (1/n,n)\}.
\end{equation}
Then for any $\epsilon\in (0,1)$, there exists a pair of positive constants $n=n(\epsilon)$ and $\Delta^1=\Delta^1(\epsilon)$ such that for each $\Delta\in (0,\Delta^1]$, we have
\begin{equation}\label{eq:7}
  \mathbb{P}(\vartheta^*_n\le T)\le \epsilon.
\end{equation}
\end{lemma}
\begin{proof}
We apply the It\^{o} formula to \eqref{sec3:eq:6} to compute 
\begin{align*}
&\mathbb{E}(H(x_{\Delta}(t\wedge \vartheta^*_n)))-H(x(0))\\
&=\mathbb{E}\int_{0}^{t\wedge\vartheta^*_n}\Big(H_x(x_{\Delta}(s))f_1^{\Delta}(\bar{x}_{\Delta}(s))+\frac{1}{2}H_{xx}(x_{\Delta}(s))\bar{\varphi}_{\Delta}(s)g_1^{\Delta}(\bar{x}_{\Delta}(s))^2\Big)ds\\
&\le \mathcal{J}_3+\mathcal{J}_4+\mathcal{J}_5
\end{align*}
where,
\begin{align*}
\mathcal{J}_3&=\mathbb{E}\int_{0}^{t\wedge\vartheta^*_n}\Big(H_x(x_{\Delta}(s))f_1^{\Delta}(x_{\Delta}(s))+\frac{1}{2}H_{xx}(x_{\Delta}(s))\varphi_{\Delta}(s)g_1^{\Delta}(x_{\Delta}(s))^2\Big)\\
\mathcal{J}_4&=\mathbb{E}\int_{0}^{t\wedge\vartheta^*_n}H_x(x_{\Delta}(s))\Big(f_1^{\Delta}(\bar{x}_{\Delta}(s))-f_1^{\Delta}(x_{\Delta}(s))\Big)ds\\
\mathcal{J}_5&=\mathbb{E}\int_{0}^{t\wedge\vartheta^*_n}\frac{1}{2}H_{xx}(x_{\Delta}(s))\Big(\bar{\varphi}_{\Delta}(s)g_1^{\Delta}(\bar{x}_{\Delta}(s))^2-\varphi_{\Delta}(s)g_1^{\Delta}(x_{\Delta}(s))^2\Big)ds.
\end{align*}
So, By \eqref{sec2:eq:3} and \eqref{sec3:eq:1}, we can find a constant $K_9$ such that
\begin{align*}
\mathcal{J}_3&\le \mathbb{E}\int_{0}^{t\wedge\vartheta^*_n}LH(x_{\Delta}(s),\varphi_{\Delta}(s))ds\\
&\le K_9T.
\end{align*}
By the definition of the truncated functions, we note for $s\in [0,t\wedge\vartheta^*_n]$,
\begin{equation}\label{sec5:eq:recall1}
f_1^{\Delta}(x_{\Delta}(s))=f_1(x_{\Delta}(s))\text{ and } g_1^{\Delta}(x_{\Delta}(s))=g_1(x_{\Delta}(s)).
\end{equation}
So by Lemma \ref{sec4:eq:L1}, we have
where,
\begin{align*}
\mathcal{J}_4&\le \mathbb{E}\int_{0}^{t\wedge\vartheta^*_n}H_x(x_{\Delta}(s))\vert f_1(\bar{x}_{\Delta}(s))-f_1(x_{\Delta}(s))\vert ds\\
&\le \mathbb{E}\int_{0}^{t\wedge\vartheta^*_n}K_nH_x(x_{\Delta}(s))\vert \bar{x}_{\Delta}(s)-x_{\Delta}(s)\vert ds.
\end{align*}
Similarly,
\begin{align*}
\mathcal{J}_5&=\mathbb{E}\int_{0}^{t\wedge\vartheta^*_n}\frac{1}{2}H_{xx}(x_{\Delta}(s))\Big(\bar{\varphi}_{\Delta}(s)g_1(\bar{x}_{\Delta}(s))^2-\varphi(s)g_1(x_{\Delta}(s))^2\Big)ds\\
&=\mathbb{E}\int_{0}^{t\wedge\vartheta^*_n}\frac{1}{2}H_{xx}(x_{\Delta}(s))\Big(\bar{\varphi}_{\Delta}(s)g_1(\bar{x}_{\Delta}(s))^2-\bar{\varphi}_{\Delta}(s)g_1(x(s))^2\Big)ds\\
&+\mathbb{E}\int_{0}^{t\wedge\vartheta^*_n}\frac{1}{2}H_{xx}(x_{\Delta}(s))\Big(\bar{\varphi}_{\Delta}(s)g_1(x(s))^2-\varphi(s)g_1(x_{\Delta}(s))^2\Big)ds\\
&\le \mathbb{E}\int_{0}^{t\wedge\vartheta^*_n}\frac{\bar{\varphi}_{\Delta}(s)}{2}H_{xx}(x_{\Delta}(s))\vert g_1(\bar{x}_{\Delta}(s))-g_1(x_{\Delta}(s))\vert \vert g_1(\bar{x}_{\Delta}(s))+g_1(x_{\Delta}(s))\vert ds\\
&+\mathbb{E}\int_{0}^{t\wedge\vartheta^*_n}\frac{1}{2}H_{xx}(x_{\Delta}(s))g_1(x(s))^2\vert\bar{\varphi}_{\Delta}(s)-\varphi(s)\vert ds.
\end{align*}
Noting from \eqref{sec4:eq:5} that $x_{\Delta}(s),\bar{x}_{\Delta}(s)\in[1/n,n]$ for $s\in[0,t\wedge\vartheta^*_n]$, we have $ g_1(\bar{x}_{\Delta}(s))\vee g_1(x_{\Delta}(s))\le \nu(n)$. So by Lemma \ref{sec4:eq:L1}, we obtain
\begin{align*}
\mathcal{J}_5&\le \mathbb{E}\int_{0}^{t\wedge\vartheta^*_n}\bar{\varphi}_{\Delta}(s)H_{xx}(x_{\Delta}(s))\vert g_1(\bar{x}_{\Delta}(s))-g_1(x_{\Delta}(s))\vert ds\\
&+\mathbb{E}\int_{0}^{t\wedge\vartheta^*_n}\frac{\nu(n)^2}{2}H_{xx}(x_{\Delta}(s))\vert\bar{\varphi}_{\Delta}(s)-\varphi(s)\vert ds\\
&\le \mathbb{E}\int_{0}^{t\wedge\vartheta^*_n}K_n\bar{\varphi}_{\Delta}(s)H_{xx}(x_{\Delta}(s))\vert \bar{x}_{\Delta}(s)-x_{\Delta}(s)\vert ds\\
&+\mathbb{E}\int_{0}^{t\wedge\vartheta^*_n}\frac{\nu(n)^2}{2}H_{xx}(x_{\Delta}(s))\vert\bar{\varphi}_{\Delta}(s)-\varphi(s)\vert ds.
\end{align*}
Combining $\mathcal{J}_3$, $\mathcal{J}_4$ and $\mathcal{J}_5$, we then have
\begin{align*}
&\mathbb{E}(H(x_{\Delta}(t\wedge \vartheta^*_n)))\\
&\le H(x(0))+K_9T+\mathbb{E}\int_{0}^{t\wedge\vartheta^*_n}K_nH_x(x_{\Delta}(s))\vert \bar{x}_{\Delta}(s)-x_{\Delta}(s)\vert ds\\
&+\mathbb{E}\int_{0}^{t\wedge\vartheta^*_n}\frac{\nu(n)^2}{2}H_{xx}(x_{\Delta}(s))\vert\bar{\varphi}_{\Delta}(s)-\varphi(s)\vert ds\\
&+ \mathbb{E}\int_{0}^{t\wedge\vartheta^*_n}K_n\bar{\varphi}_{\Delta}(s)H_{xx}(x_{\Delta}(s))\vert \bar{x}_{\Delta}(s)-x_{\Delta}(s)\vert ds\\
&\le  H(x(0))+K_9T+\iota_2\int_{0}^{T}\mathbb{E}\vert\bar{\varphi}_{\Delta}(s)-\varphi(s)\vert ds\\
&+\iota_3\int_{0}^{T}\mathbb{E}(\vert\bar{\varphi}_{\Delta}(s)\vert\vert \bar{x}_{\Delta}(s)-x_{\Delta}(s)\vert )ds
\end{align*}
where
\begin{equation*}
\iota_2=\max_{1/n\le x\le n}\Big(\frac{\nu(n)^2}{2}H_{xx}(x)\Big)
\end{equation*}
and
\begin{equation*}
\iota_3=\max_{1/n\le x\le n}\Big(K_nH_x(x)+K_nH_{xx}(x)\Big).
\end{equation*}
So by the Young inequality and Lemmas \ref{sec5:eq:L1}, \ref{sec5:eq:L1**} and \ref{sec5:eq:L1***}, we now have 
\begin{align*}
\mathbb{E}(H(x_{\Delta}(t\wedge \vartheta^*_n)))
&\le  H(x(0))+K_9T+\iota_2\int_{0}^{T}\mathbb{E}\vert\bar{\varphi}_{\Delta}(s)-\varphi(s)\vert ds\\
&+\iota_3\int_{0}^{T}\mathbb{E}(\vert\bar{\varphi}_{\Delta}(s)\vert^2)^{\frac{1}{2}}(\vert \bar{x}_{\Delta}(s)-x_{\Delta}(s)\vert^2 )^{\frac{1}{2}}ds\\
&\le  H(x(0))+K_9T+\iota_2\int_{0}^{T}\mathbb{E}\vert\bar{\varphi}_{\Delta}(s)-\varphi(s)\vert ds\\
&+\frac{\iota_3}{2}\int_{0}^{T}\mathbb{E}\vert\bar{\varphi}_{\Delta}(s)\vert^2ds+\frac{\iota_3}{2}\int_{0}^{T}(\mathbb{E}\vert \bar{x}_{\Delta}(s)-x_{\Delta}(s)\vert^p)^\frac{2}{p} ds\\
&\le  H(x(0))+K_9T+\iota_2c_pT\Delta^{p/2} (h(\Delta))^p\\
&+\frac{\iota_3}{2}\big(C_p\Delta^{p/2} (h(\Delta))^p)^\frac{2}{p}T+\frac{\iota_3}{2}\int_{0}^{T}\mathbb{E}\vert\bar{\varphi}_{\Delta}(s)\vert^p ds\\
&\le  H(x(0))+K_9T+\iota_2c_pT\Delta^{p/2} (h(\Delta))^p\\
&+\frac{\iota_3}{2}C_p\Delta(h(\Delta))^2T+\frac{\iota_3}{2}\int_{0}^{T}\big(\sup_{0\le u\le s}\mathbb{E}\vert\bar{\varphi}_{\Delta}(u)\vert^p\big)^\frac{2}{p} ds\\
&\le  H(x(0))+K_9T+\iota_2c_pT\Delta^{p/2} (h(\Delta))^p\\
&+\frac{\iota_3}{2}C_p\Delta(h(\Delta))^2T+\frac{\iota_3}{2}c_4^\frac{2}{p}T.
\end{align*}
This implies 
\begin{equation}\label{eq:34}
\mathbb{P}(\vartheta^*_n\leq T)\leq \frac{H(x(0))+K_9T+\iota_2c_pT\Delta^{p/2} (h(\Delta))^p+\frac{\iota_3}{2}C_p\Delta(h(\Delta))^2T+\frac{\iota_3}{2}c_4^\frac{2}{p}T}{H(1/n)\wedge H(n)}.
\end{equation}
For any $\epsilon\in(0,1)$, we may select sufficiently large $n$ such that
\begin{equation}\label{eq:35}
\frac{H(x(0))+K_9T+\frac{\iota_3}{2}c_4^\frac{2}{p}T}{H(1/n)\wedge H(n)}\leq \frac{\epsilon}{2}
\end{equation}
and sufficiently small of each step size $\Delta\in (0,\Delta^1]$ such that
\begin{equation}\label{eq:36}
\frac{\iota_2c_pT\Delta^{p/2} (h(\Delta))^p+\frac{\iota_3}{2}C_p\Delta(h(\Delta))^2T}{H(1/n)\wedge H(n)}\leq \frac{\epsilon}{2}.
\end{equation}
We now combine \eqref{eq:35} and \eqref{eq:36} to get the required assertion.
\end{proof}
\begin{lemma}\label{eq:l8}
Let Assumption \ref{sec3:assump:1} hold.  Set
\begin{equation*}
\upsilon^*_n=\varrho_{mn}\wedge\vartheta_n\wedge\vartheta^*_n ,
\end{equation*}
where $\varrho_{mn}$, $\vartheta_n$ and $\vartheta^*_n$ are \eqref{sec3:eq:5}, \eqref{sec5:eq:8} and \eqref{eq:7*} respectively. Then for any $p\geq 2$, $T> 0$, we have
\begin{equation}\label{eq:37}
\mathbb{E}\Big( \sup_{0\leq t \leq T}|x_{\Delta}(t \wedge \upsilon^*_n)-x(t \wedge \upsilon^*_n)|^p  \Big)\le \mathcal{K}_2\Delta^{p(1/2\wedge 1/4\wedge 1/8)} (h(\Delta))^{p(1/2\wedge 1)} 
\end{equation}
for any sufficiently large $n$ and any $\Delta\in (0,\Delta^*]$, where $\mathcal{K}_2$ is a constant independent of $\Delta$. Consequently, we have
\begin{equation}\label{eq:38}
\lim_{\Delta\rightarrow 0}\mathbb{E}\Big( \sup_{0\leq t \leq T}|x_{\Delta}(t \wedge \upsilon^*_n)-x(t \wedge \upsilon^*_n)|^p \Big)=0.
\end{equation}
\end{lemma}
\begin{proof}
It follows from \eqref{sec2:eq:1} and \eqref{sec4:eq:13} that
\begin{align*}
\mathbb{E}\Big(\sup_{0\leq t \leq t_1}|x_{\Delta}(t\wedge \upsilon^*_n)-x(t\wedge \upsilon^*_n)|^p\Big)\le \mathcal{J}_6+\mathcal{J}_7,
\end{align*}
 where
\begin{align*} 
\mathcal{J}_6&=2^{p-1}\Big( \mathbb{E}\Big|\int_{0}^{t_1\wedge \upsilon^*_n}(f_1^{\Delta}(\bar{x}_{\Delta}(s))-f_1(x(s)))ds\Big|^p\Big)\\
\mathcal{J}_7&=2^{p-1}\Big(\mathbb{E}(\sup_{0\leq t \leq t_1}\Big|\int_{0}^{t_1\wedge \upsilon^*_n}(\sqrt{\vert\bar{\varphi}_{\Delta}(s)\vert} g_1^{\Delta}(\bar{x}_{\Delta}(s))-\sqrt{\vert\varphi(s)\vert} g_1(x(s)))dB(s)\Big|^p)\Big).
\end{align*}
So by the H\"older inequality, \eqref{sec4:eq:1} and \eqref{sec5:eq:recall1}, we have 
\begin{align*} 
\mathcal{J}_6&\le 2^{p-1}T^{p-1}\Big( \mathbb{E}\int_{0}^{t_1\wedge \upsilon^*_n}|f_1^{\Delta}(\bar{x}_{\Delta}(s))-f_1(x(s))|^pds\Big)\\
&\le 2^{p-1}T^{p-1}\Big(\mathbb{E}\int_{0}^{t_1\wedge \upsilon^*_n}|f_1(\bar{x}_{\Delta}(s))-f_1(x(s))|^pds\Big)\\
&\le 2^{p-1}T^{p-1}K_n\mathbb{E}\int_{0}^{t_1\wedge \upsilon^*_n}|\bar{x}_{\Delta}(s)-x(s)|^pds\\
&\le 2^{2(p-1)}T^{p-1}K_n\mathbb{E}\int_{0}^{t_1\wedge \upsilon^*_n}|\bar{x}_{\Delta}(s)-x_{\Delta}(s)|^pds\\
&+2^{2(p-1)}T^{p-1}K_n\mathbb{E}\int_{0}^{t_1\wedge \upsilon^*_n}|x_{\Delta}(s)-x(s)|^pds\\
&\le 2^{2(p-1)}T^{p-1}K_n\int_{0}^{T}\mathbb{E}|\bar{x}_{\Delta}(s)-x_{\Delta}(s)|^pds\\
&+2^{2(p-1)}T^{p-1}K_n\int_{0}^{t_1}\sup_{0\le t\le s}\mathbb{E}|x_{\Delta}(t\wedge \upsilon^*_n)-x(t\wedge \upsilon^*_n)|^pds.
\end{align*}
By the Burkholder-Davis-Gundy inequality and \eqref{sec5:eq:recall1}, we also have 
\begin{align*}
\mathcal{J}_7&\le 2^{p-1}T^{\frac{p-2}{2}} \bar{C}_p\Big(\mathbb{E}\int_{0}^{t_1\wedge \upsilon^*_n}\vert \sqrt{\vert\bar{\varphi}_{\Delta}(s)\vert}g_1^{\Delta}(\bar{x}_{\Delta}(s))-\sqrt{\vert\varphi(s)\vert}g_1(x(s))\vert^pds \Big)\\
&\le 2^{p-1}T^{\frac{p-2}{2}} \bar{C}_p\Big(\mathbb{E}\int_{0}^{t_1\wedge \upsilon^*_n}\vert \sqrt{\vert\bar{\varphi}_{\Delta}(s)\vert}g_1(\bar{x}_{\Delta}(s))-\sqrt{\vert\varphi(s)\vert}g_1(x(s))\vert^pds\Big),
\end{align*}
where $\bar{C}_p$ is a positive constant. By elementary inequality, we now have 
\begin{align*}
\mathcal{J}_7&\le 2^{2(p-1)}T^{\frac{p-2}{2}} \bar{C}_p\Big(\mathbb{E}\int_{0}^{t_1\wedge \upsilon^*_n}\vert \sqrt{\vert\bar{\varphi}_{\Delta}(s)\vert}g_1(\bar{x}_{\Delta}(s))-\sqrt{\vert\bar{\varphi}_{\Delta}(s)\vert}g_1(x_{\Delta}(s))\vert^pds\Big)\\
&+ 2^{2(p-1)}T^{\frac{p-2}{2}} \bar{C}_p\Big(\mathbb{E}\int_{0}^{t_1\wedge \upsilon^*_n}\vert \sqrt{\vert\bar{\varphi}_{\Delta}(s)\vert}g_1(x_{\Delta}(s))-\sqrt{\vert\varphi(s)\vert}g_1(x(s))\vert^pds\Big)\\
&\le 2^{2(p-1)}T^{\frac{p-2}{2}} \bar{C}_p\Big(\mathbb{E}\int_{0}^{t_1\wedge \upsilon^*_n}\vert \sqrt{\vert\bar{\varphi}_{\Delta}(s)\vert}g_1(\bar{x}_{\Delta}(s))-\sqrt{\varphi(s)}g_1(\bar{x}_{\Delta}(s))\vert^pds\Big)\\
&+ 2^{2(p-1)}T^{\frac{p-2}{2}} \bar{C}_p\Big(\mathbb{E}\int_{0}^{t_1\wedge \upsilon^*_n}\vert \sqrt{\vert\varphi(s)\vert}g_1(\bar{x}_{\Delta}(s))-\sqrt{\vert\varphi(s)\vert}g_1(x(s))\vert^pds\Big)\\
&\le 2^{2(p-1)}T^{\frac{p-2}{2}} \bar{C}_p\Big(\mathbb{E}\int_{0}^{t_1\wedge \upsilon^*_n}g_1(\bar{x}_{\Delta}(s))^p\vert \sqrt{\vert\bar{\varphi}_{\Delta}(s)\vert}-\sqrt{\vert\varphi(s)\vert}\vert^pds\Big)\\
&+ 2^{2(p-1)}T^{\frac{p-2}{2}} \bar{C}_p\Big(\mathbb{E}\int_{0}^{t_1\wedge \upsilon^*_n}\vert\varphi(s)\vert^{\frac{p}{2}}\vert g_1(\bar{x}_{\Delta}(s))-g_1(x(s))\vert^pds\Big).
\end{align*}
It follows from \eqref{sec3:eq:3} and \eqref{sec4:eq:5} that
\begin{align*}
\mathcal{J}_7
&\le 2^{2(p-1)}T^{\frac{p-2}{2}} \bar{C}_p\nu(n)^p\Big(\mathbb{E}\int_{0}^{t_1\wedge \upsilon^*_n}\vert \sqrt{\vert\bar{\varphi}_{\Delta}(s)\vert}-\sqrt{\vert\varphi(s)\vert}\vert^pds\Big)\\
&+ 2^{2(p-1)}T^{\frac{p-2}{2}} \bar{C}_p n^{\frac{p}{2}}\Big(\mathbb{E}\int_{0}^{t_1\wedge \upsilon^*_n}\vert g_1(\bar{x}_{\Delta}(s))-g_1(x(s))\vert^pds\Big)\\
&\le 2^{2(p-1)}T^{\frac{p-2}{2}} \bar{C}_p\nu(n)^p\Big(\mathbb{E}\int_{0}^{t_1\wedge \upsilon^*_n}\vert \bar{\varphi}_{\Delta}(s)-\varphi(s)\vert^{\frac{p}{2}}ds\Big)\\
&+ 2^{2(p-1)}T^{\frac{p-2}{2}} \bar{C}_p n^{\frac{p}{2}}\Big(\mathbb{E}\int_{0}^{t_1\wedge \upsilon^*_n}\vert g_1(\bar{x}_{\Delta}(s))-g_1(x(s))\vert^pds\Big).
\end{align*}
Then, by elementary inequality and \eqref{sec4:eq:1}, we have
\begin{align*}
\mathcal{J}_7&\le 2^{2(p-1)}2^{\frac{p}{2}-1}T^{\frac{p-2}{2}} \bar{C}_p\nu(n)^p\Big(\mathbb{E}\int_{0}^{t_1\wedge \upsilon^*_n}\vert \bar{\varphi}_{\Delta}(s)-\varphi_{\Delta}(s)\vert^{\frac{p}{2}}ds\Big)\\
&+2^{2(p-1)}2^{\frac{p}{2}-1}T^{\frac{p-2}{2}} \bar{C}_p\nu(n)^p\Big(\mathbb{E}\int_{0}^{t_1\wedge \upsilon^*_n}\vert \varphi_{\Delta}(s)-\varphi(s)\vert^{\frac{p}{2}}ds\Big)\\
&+ 2^{3(p-1)}T^{\frac{p-2}{2}} \bar{C}_p n^{\frac{p}{2}}K_n\Big(\mathbb{E}\int_{0}^{t_1\wedge \upsilon^*_n}\vert \bar{x}_{\Delta}(s)-x_{\Delta}(s)\vert^pds\Big)\\
&+ 2^{3(p-1)}T^{\frac{p-2}{2}} \bar{C}_p n^{\frac{p}{2}}K_n\Big(\mathbb{E}\int_{0}^{t_1\wedge \upsilon^*_n}\vert x_{\Delta}(s)-x(s)\vert^pds\Big)\\
&\le 2^{2(p-1)}2^{\frac{p}{2}-1}T^{\frac{p-2}{2}} \bar{C}_p\nu(n)^p\int_{0}^{T	}\Big(\mathbb{E}\vert \bar{\varphi}_{\Delta}(s)-\varphi_{\Delta}(s)\vert^p\Big)^{\frac{1}{2}}ds\\
&+2^{2(p-1)}2^{\frac{p}{2}-1}T^{\frac{p-2}{2}} \bar{C}_p\nu(n)^p\int_{0}^{t_1}\Big(\sup_{0\le t\le s} \mathbb{E}\vert \varphi_{\Delta}(t\wedge \upsilon^*_n)-\varphi(t\wedge \upsilon^*_n)\vert^p\Big)^{\frac{1}{2}}ds\\
&+ 2^{3(p-1)}T^{\frac{p-2}{2}} \bar{C}_p n^{\frac{p}{2}}K_n\int_{0}^{T}\mathbb{E}\vert \bar{x}_{\Delta}(s)-x_{\Delta}(s)\vert^pds\\
&+ 2^{3(p-1)}T^{\frac{p-2}{2}} \bar{C}_p n^{\frac{p}{2}}K_n\int_{0}^{t_1}\sup_{0\le t\le s} \mathbb{E}\vert x_{\Delta}(t\wedge \upsilon^*_n)-x(t\wedge \upsilon^*_n)\vert^pds.
\end{align*}
So by combining $\mathcal{J}_6$ and $\mathcal{J}_7$, we now get
\begin{align*}
&\mathbb{E}\Big(\sup_{0\leq t \leq t_1}|x_{\Delta}(t\wedge \upsilon^*_n)-x(t\wedge \upsilon^*_n)|^p\Big)\le 2^{2(p-1)}T^{p-1}K_n\int_{0}^{T}\mathbb{E}|\bar{x}_{\Delta}(s)-x_{\Delta}(s)|^pds\\
&+2^{2(p-1)}T^{p-1}K_n\int_{0}^{t_1}\sup_{0\le t\le s}\mathbb{E}|x_{\Delta}(t\wedge \upsilon^*_n)-x(t\wedge \upsilon^*_n)|^pds\\
&+ 2^{2(p-1)}2^{\frac{p}{2}-1}T^{\frac{p-2}{2}} \bar{C}_p\nu(n)^p\int_{0}^{T	}\Big(\mathbb{E}\vert \bar{\varphi}_{\Delta}(s)-\varphi_{\Delta}(s)\vert^p\Big)^{\frac{1}{2}}ds\\
&+2^{2(p-1)}2^{\frac{p}{2}-1}T^{\frac{p-2}{2}} \bar{C}_p\nu(n)^p\int_{0}^{t_1}\Big(\sup_{0\le t\le s} \mathbb{E}\vert \varphi_{\Delta}(t\wedge \upsilon^*_n)-\varphi(t\wedge \upsilon^*_n)\vert^p\Big)^{\frac{1}{2}}ds\\
&+ 2^{3(p-1)}T^{\frac{p-2}{2}} \bar{C}_p n^{\frac{p}{2}}K_n\int_{0}^{T}\mathbb{E}\vert \bar{x}_{\Delta}(s)-x_{\Delta}(s)\vert^pds\\
&+ 2^{3(p-1)}T^{\frac{p-2}{2}} \bar{C}_p n^{\frac{p}{2}}K_n\int_{0}^{t_1}\sup_{0\le t\le s} \mathbb{E}\vert x_{\Delta}(t\wedge \upsilon^*_n)-x(t\wedge \upsilon^*_n)\vert^pds.
\end{align*}
So, by Lemmas \ref{sec5:eq:L1**}, \ref{sec5:eq:L1**} and \ref{eq:L3}, we hence have 
\begin{align*}
&\mathbb{E}\Big(\sup_{0\leq t \leq t_1}|x_{\Delta}(t\wedge \upsilon^*_n)-x(t\wedge \upsilon^*_n)|^p\Big)\\
&\le 2^{2(p-1)}2^{\frac{p}{2}-1}T^{\frac{p-2}{2}}\bar{C}_p\nu(n)^pT(c_p\Delta^{p/2}(h(\Delta))^p)^{\frac{1}{2}}\\
&+2^{2(p-1)}2^{\frac{p}{2}-1}T^{\frac{p-2}{2}} \bar{C}_p\nu(n)^pT(\mathcal{K}_1\Delta^{p/4})^{\frac{1}{2}}\\
&+(2^{2(p-1)}T^{p-1}K_n+ 2^{3(p-1)}T^{\frac{p-2}{2}}\bar{C}_p n^{\frac{p}{2}}K_n)C_p\Delta^{p/2} (h(\Delta))^p\\
&+(2^{2(p-1)}T^{p-1}K_n+2^{3(p-1)}T^{\frac{p-2}{2}} \bar{C}_p n^{\frac{p}{2}}K_n)\int_{0}^{t_1}\sup_{0\le t\le s} \mathbb{E}\vert x_{\Delta}(t\wedge \upsilon^*_n)-x(t\wedge \upsilon^*_n)\vert^pds.
\end{align*}
In particular, we have 
\begin{align*}
\mathbb{E}\Big(\sup_{0\leq t \leq t_1}|x_{\Delta}(t\wedge \upsilon^*_n)-x(t\wedge \upsilon^*_n)|^p\Big)&\le (\varpi_3+\varpi_4+\varpi_5)\Delta^{p(1/2\wedge 1/4\wedge 1/8)} (h(\Delta))^{p(1/2\wedge 1)}\\
&+ \varpi_6\int_{0}^{t_1}\sup_{0\le t\le s} \mathbb{E}\vert x_{\Delta}(t\wedge \upsilon^*_n)-x(t\wedge \upsilon^*_n)\vert^pds
\end{align*}
where
\begin{align*}
\varpi_3&=2^{2(p-1)}2^{\frac{p}{2}-1}T^{\frac{p-2}{2}}\bar{C}_p\nu(n)^pTc_p^{\frac{1}{2}}\\
\varpi_4&=2^{2(p-1)}2^{\frac{p}{2}-1}T^{\frac{p-2}{2}} \bar{C}_p\nu(n)^pT\mathcal{K}_1^{\frac{1}{2}}\\
\varpi_5&=(2^{2(p-1)}T^{p-1}K_n+ 2^{3(p-1)}T^{\frac{p-2}{2}}\bar{C}_p n^{\frac{p}{2}}K_n)C_p\Delta^{p/2} (h(\Delta))^p\\
\varpi_6&=2^{2(p-1)}T^{p-1}K_n+2^{3(p-1)}T^{\frac{p-2}{2}} \bar{C}_p n^{\frac{p}{2}}K_n.
\end{align*}
The Gronwall inequality shows
\begin{align*}
\mathbb{E}\Big(\sup_{0\leq t \leq t_1}|x_{\Delta}(t\wedge \upsilon^*_n)-x(t\wedge \upsilon^*_n)|^p\Big)\le \mathcal{K}_2\Delta^{p(1/2\wedge 1/4\wedge 1/8)} (h(\Delta))^{p(1/2\wedge 1)}
\end{align*}
as the required assertion, where
\begin{equation*}
\mathcal{K}_2=(\varpi_1+\varpi_2+\varpi_3)e^{\varpi_4}.
\end{equation*}
\end{proof}
The following lemma shows that the truncated EM solutions converge strongly to the exact solution without the stopping time.
\begin{theorem}\label{eq:thrm2}
Let Assumptions \ref{sec3:assump:1} hold. Then for any $p\ge 2$, we have
\begin{equation}\label{eq:42}
\lim_{\Delta\rightarrow 0}\mathbb{E}\Big( \sup_{0\leq t \leq T}|x_{\Delta}(t)-x(t)|^p \Big)=0
\end{equation}
and consequently
\begin{equation}\label{eq:43}
\lim_{\Delta\rightarrow 0}\mathbb{E}\Big( \sup_{0\leq t \leq T}|\bar{x}_{\Delta}(t)-x(t)|^p \Big)=0.
\end{equation}
\end{theorem}
\begin{proof}
Let $\varrho_{mn}$, $\vartheta^*_n$ and $\upsilon^*_n$ be the same as before. Now set
\begin{equation*}
e_{\Delta}(t)=x_{\Delta}(t)-x(t).
\end{equation*}
For any arbitrarily $ \delta >0$, we derive from the Young inequality that
\begin{align}\label{eq:44}
\mathbb{E}\Big(\sup_{0\leq t \leq T}|e_{\Delta}(t)|^p\Big)&=\mathbb{E}\Big(\sup_{0\leq t \leq T}|e_{\Delta}(t)|^p1_{\{\varrho_{mn}>T \text{ and }\vartheta^*_n>T\}}\Big)\\
&+\mathbb{E}\Big(\sup_{0\leq t \leq T}|e_{\Delta}(t)|^p1_{\{\varrho_{mn}\le T \text{ or }\vartheta^*_n\le T\}}\Big)\nonumber\\
&\le \mathbb{E}\Big(\sup_{0\leq t \leq T}|e_{\Delta}(t)|^p1_{\{\upsilon^*_n>T\}}\Big)+\frac{\delta}{2}\mathbb{E}\Big(\sup_{0\leq t \leq T}|e_{\Delta}(t)|^{2p}\Big)\nonumber\\
&+\frac{1}{2\delta}\mathbb{P}(\varrho_{mn}\leq T \text{ or }\vartheta^*_n \leq T).
\end{align}
Then for $ p\ge 2$, Lemmas \ref{sec3:eq:L2} and \ref{sec5:eq:L2} give us
\begin{align}\label{eq:45}
 \mathbb{E}\Big(\sup_{0\leq t \leq T} |e_{\Delta}(t)|^{2p}\Big)&\leq 2^{2p}\mathbb{E}\Big(\sup_{0\leq t \leq T}|x(t)|^p\vee \sup_{0\leq t \leq T}|x_{\Delta}(t)|^p\Big)^2
 \nonumber\\
&\le 2^{2p}(c_2\vee c_5)^2.
\end{align}
By Lemmas \ref{sec3:eq:00} and \ref{eq:L3}, we have
\begin{equation}\label{eq:46}
\mathbb{P}(\upsilon^*_n\leq T)\le \mathbb{P}(\varrho_{mn}\leq T) +\mathbb{P}(\vartheta^*_n\leq T).
\end{equation}
Also, by Lemma \ref{eq:l8}, we get
\begin{equation}\label{eq:47}
\mathbb{E}\Big(\sup_{0\leq t \leq T}|e_{\Delta}(t)|^p1_{\{\varsigma_{\Delta,n}>T\}}\Big)\le \mathcal{K}_2\Delta^{p(1/2\wedge 1/4\wedge 1/8)} (h(\Delta))^{p(1/2\wedge 1)}.
\end{equation}
Therefore, we substitute \eqref{eq:45}, \eqref{eq:46} and \eqref{eq:47} into \eqref{eq:44} to have
\begin{align*}
\mathbb{E}\Big(\sup_{0\leq t \leq T}|e_{\Delta}(t)|^p\Big)&\leq \frac{2^{2p}(c_2\vee c_5)^2\delta}{2}+ \mathcal{K}_2\Delta^{p(1/2\wedge 1/4\wedge 1/8)} (h(\Delta))^{p(1/2\wedge 1)}\\
&+ \mathbb{P}(\varrho_{mn}\leq T) +\mathbb{P}(\vartheta^*_n\leq T).
\end{align*}
Given $\epsilon\in (0,1)$, we can choose $\delta$ so that
\begin{equation}\label{eq:48}
\frac{2^{2p}(c_2\vee c_5)^2\delta}{2} \le \frac{\epsilon}{4}.
\end{equation}
Furthermore, for any given $\epsilon\in (0,1)$, there exists $n_o$ such that for $n\geq n_o$, we may choose $\delta$ to obtain
\begin{equation}\label{eq:49}
\frac{1}{2\delta}\mathbb{P}(\varrho_{mn}\leq T)\le \frac{\epsilon}{4}
\end{equation}
and then choose $n(\epsilon)\leq n_o$ such that for $\Delta\in (0,\Delta^1]$, we have
\begin{equation}\label{eq:50}
\frac{1}{2\delta}\mathbb{P}(\vartheta^*_n\leq T)\le \frac{\epsilon}{4}.
\end{equation}
Lastly, we may choose $\Delta\in (0,\Delta^1]$ sufficiently small for $\epsilon\in (0,1)$ such that
\begin{equation}\label{eq:51}
\mathcal{K}_2\Delta^{p(1/2\wedge 1/4\wedge 1/8)} (h(\Delta))^{p(1/2\wedge 1)}\le \frac{\epsilon}{4}.
\end{equation}
We then \eqref{eq:48}, \eqref{eq:49}, \eqref{eq:50} and \eqref{eq:51},  to have
\begin{align*}
\mathbb{E}\Big(\sup_{0\leq t \leq T}|x_{\Delta}(t)-x(t)|^p\Big)\le \epsilon.
\end{align*}
as desired. By Lemma \ref{sec5:eq:L1***}, we also obtain \eqref{eq:43} by letting $\Delta\rightarrow 0$.
\end{proof}
\section{Numerical application}\label{sect6}
We now provide numerical demonstrations to support the theoretical result.
\subsection{Simulation}
In what follows, let us consider the following form of SDE \eqref{sec1:eq:6}
\begin{equation}\label{eq:sm1}
 dx(t)=2(1-x(t)^{5})dt+3\sqrt{\vert\varphi(t)\vert}x(t)^{5/4}dB_1(t),
\end{equation}
with initial data $x_0=0.2$, where $\varphi(t)$ is driven by SDE \eqref{sec1:eq:7} of the form
\begin{equation}\label{eq:sm2}
 d\varphi(t)=2(2-\varphi(t)^{2})dt+0.5\varphi(t)^{3/2}dB_2(t)
\end{equation}
with initial data $\varphi_0=2$. Apparently, the coefficient terms $f_1(x)=2(1-x^{5})$, $g_1(x)=3x^{5/4}$, $f_2(x)=2(2-x^{2})$ and $g_2(\varphi)=0.5\varphi^{3/2}$ of SDE \eqref{eq:sm1} and SDE \eqref{eq:sm2} are locally Lipschitz continuous. Moreover, we observe that
\begin{equation*}
\sup_{\vert x\vert \vee\vert \varphi\vert \le u}\Big(|f_1(x)|\vee|f_2(\varphi)|\vee g_1(x)\vee g_2(\varphi)\Big)\le 10.5 \nu^5,\quad \nu\ge 0,
\end{equation*}
If we choose $h(\Delta)=\Delta^{-1/2}$, then $\nu^{-1}(h(\Delta))=(\Delta/10.5)^{-1/10}$. Using a step size of $10^{-3}$, we get Monte Carlo simulated sample trajectories of SDE \eqref{eq:sm2} and SDE \eqref{eq:sm1} in Figure \ref{Fig:figure1} and Figure \ref{Fig:figure2} respectively.
\begin{figure}[!htbp]
  \centerline{\includegraphics[scale=1]{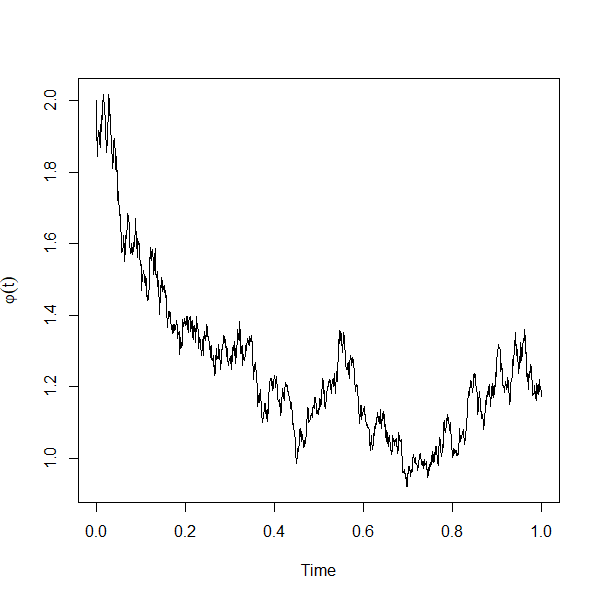}}
  \caption{Simulated sample path of $\varphi(t)$ using $\Delta=0.001$}
  \label{Fig:figure1}
\end{figure}
\begin{figure}[!htbp]
  \centerline{\includegraphics[scale=1]{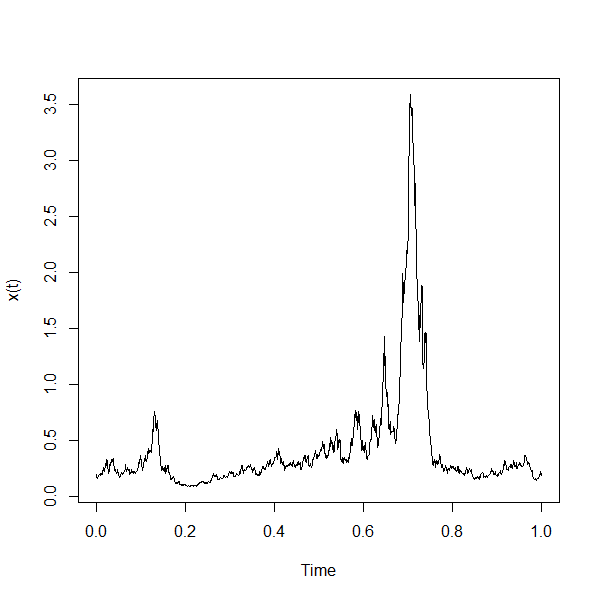}}
  \caption{Simulated sample path of $x(t)$ using $\Delta=0.001$}
  \label{Fig:figure2}
\end{figure}
\subsection{Evaluation}
In this session, we justify that the truncated EM solutions can be used to compute a barrier option with a European payoff $\mathcal{P}$. Let the asset price be the exact solution $x(T)$ to SDE \eqref{sec1:eq:6}, $\mathbb{B}$ be a fixed barrier, $T$ be an expiry date and $\Lambda$ a strike price. Then the exact payoff of a barrier option is
\begin{equation*}\label{5.1}
  \mathcal{P}(T)=\mathbb{E}\Big[ (x(T)-\Lambda)^+1_{\sup_{0\leq t\leq T }}x(t)<\mathbb{B})\Big].
\end{equation*}
Using the step process \eqref{sec4:eq:11}, we could compute the approximate payoff by
\begin{equation*}\label{5.1}
  \mathcal{P}^{\Delta}(T)=\mathbb{E}\Big[ (\bar{x}_{\Delta}(T)-\Lambda)^+1_{\sup_{0\leq t\leq T }}\bar{x}_{\Delta}(t)<\mathbb{B})\Big].
\end{equation*}
So, from Theorem \ref{eq:thrm2}, we have
\begin{equation*}
  \lim_{\Delta\rightarrow 0}|\mathcal{P}(T)-\mathcal{P}^{\Delta}(T)|=0.
\end{equation*}
See \cite{highamao,maobook} for the detailed account.


\begin{thebibliography}{99}
\bibitem{blackshole}
Black, F. and Scholes, M., 1973. The pricing of options and corporate liabilities. Journal of political economy, 81(3), pp.637-654.
\bibitem{Vasicek}
Vasicek, O., 1977. An equilibrium characterization of the term structure. Journal of financial economics, 5(2), pp.177-188.
\bibitem{cox1}
Cox, J.C., Ingersoll Jr, J.E. and Ross, S.A., 1985. A Theory of the Term Structure of Interest Rates. Econometrica: Journal of the Econometric Society, pp.385-407.
\bibitem{lewis}
Lewis, A.L., 2000. Option Valuation Under Stochastic Volatility. Finance Press, California.
\bibitem{Chan}
Chan, K.C., Karolyi, G.A., Longstaff, F.A. and Sanders, A.B., 1992. An empirical comparison of alternative models of the short‐term interest rate. The journal of finance, 47(3), pp.1209-1227.
\bibitem{Nowman}
Nowman, K.B., 1997. Gaussian estimation of single‐factor continuous time models of the term structure of interest rates. The journal of Finance, 52(4), pp.1695-1706.
\bibitem{Ahn}
Ahn, D.H. and Gao, B., 1999. A parametric nonlinear model of term structure dynamics. The Review of Financial Studies, 12(4), pp.721-762.
\bibitem{Yang}
Yang, H., Wu, F., Kloeden, P.E. and Mao, X., 2020. The truncated Euler–Maruyama method for stochastic differential equations with Hölder diffusion coefficients. Journal of Computational and Applied Mathematics, 366, p.112379.
\bibitem{implicit}
Mao, X. and Szpruch, L., 2013. Strong convergence and stability of implicit numerical methods for stochastic differential equations with non-globally Lipschitz continuous coefficients. Journal of Computational and Applied Mathematics, 238, pp.14-28.
\bibitem{highamao}
Higham, D.J. and Mao, X., 2005. Convergence of Monte Carlo simulations involving the mean-reverting square root process. Journal of Computational Finance, 8(3), pp.35-61.
\bibitem{wu}
Wu, F., Mao, X. and Chen, K., 2008. A highly sensitive mean-reverting process in finance and the Euler–Maruyama approximations. Journal of Mathematical Analysis and Applications, 348(1), pp.540-554.
\bibitem{Dupire}
Dupire, B., 1994. Pricing with a smile. Risk, 7(1), pp.18-20.
\bibitem{CoxJ1}
Cox, J.C., 1975. Constant elasticity of variance diffusions. Standford University, Graduate School of Business.
\bibitem{CoxJ2}
Cox, J.C., 1996. The constant elasticity of variance option pricing model. Journal of Portfolio Management, p.15.
\bibitem{HullWhite}
Hull, J. and White, A., 1987. The pricing of options on assets with stochastic volatilities. The journal of finance, 42(2), pp.281-300.
\bibitem{Hagan}
Hagan, Patrick S, Deep Kumar, Andrew S Lesniewski, and Diana E Woodward. 2002. Managing Smile Risk. The Best of Wilmott 1: 249–96.
\bibitem{Heston}
Heston, S.L., 1993. A closed-form solution for options with stochastic volatility with applications to bond and currency options. The review of financial studies, 6(2), pp.327-343.
\bibitem{emma}
Coffie, E. and Mao, X., 2021. Truncated EM numerical method for generalised Ait-Sahalia-type interest rate model with delay. Journal of Computational and Applied Mathematics, 383, p.113137.
\bibitem{mao3}
Mao, X., 2015. The truncated Euler–Maruyama method for stochastic differential equations. Journal of Computational and Applied Mathematics, 290, pp.370-384.
\bibitem{maobook}
Mao, X., 2007. Stochastic differential equations and applications. 2nd ed. Chichester: Horwood Publishing Limited.
\bibitem{Baduraliya}
Baduraliya, C.H. and Mao, X., 2012. The Euler–Maruyama approximation for the asset price in the mean-reverting-theta stochastic volatility model. Computers \& Mathematics with Applications, 64(7), pp.2209-2223.
\end{thebibliography}
\end{document}